\pdfoutput=1

\documentclass[a4paper,entropy,article,submit,oneauthor]{mdpiTest5}
\usepackage{amsfonts}
\usepackage{amssymb}
\usepackage{amsmath}
\usepackage{graphicx}
\usepackage{amsthm}

\newtheorem{theorem}{Theorem}

\newtheorem{corollary}[theorem]{Corollary}

\newtheorem{definition}[theorem]{Definition}
\newtheorem{example}[theorem]{Example}

\newtheorem{lemma}[theorem]{Lemma}

\newtheorem{proposition}[theorem]{Proposition}

\begin{document}

\title{Maximum Entropy on Compact Groups}
\author{Peter Harremo\"{e}s}
\maketitle

\address{Centrum Wiskunde \& Informatica, Science Park 123, 1098 GB Amsterdam, Noord-Holland, The Netherlands \\
E-mail: P.Harremoes@cwi.nl
}

\abstract{On a compact group the Haar probability measure plays the
role of uniform distribution. The entropy and rate distortion theory for
this uniform distribution is studied. New results and simplified proofs on
convergence of convolutions on compact groups are presented and they can be
formulated as entropy increases to its maximum. Information theoretic
techniques and Markov chains play a crucial role. The convergence results
are also formulated via rate distortion functions. The rate of convergence
is shown to be exponential.}




\section{Introduction}

It is a well-known and celebrated result that the uniform distribution on a
finite set can be characterized as having maximal entropy. Jaynes used this
idea as a foundation of statistical mechanics \cite{Jaynes57}, and the
Maximum Entropy Principle has become a popular principle for statistical
inference \cite{Topsoe93,Jaynes,Kapur,GruDawid03,Topsoe79,HarTop01,Jaynes03}%
. Often it is used as a method to get prior distributions. On a finite set,
for any distributions $P$ we have $H(P)=H(U)-D(P\Vert U)$ where $H$ is the
Shannon entropy, $D$ is information divergence, and $U$ is the uniform
distribution. Thus, maximizing $H(P)$ is equivalent to minimizing $D(P\Vert
U)$. Minimization of information divergence can be justified by the
conditional limit theorem by Csisz\'{a}r \cite[Theorem 4]{Csiszar84}. So if
we have a good reason to use the uniform distribution as prior distribution
we automatically get a justification of the Maximum Entropy Principle. The
conditional limit theorem cannot justify the use of the uniform distribution
itself, so we need something else. Here we shall focus on symmetry.

\begin{example}
A die has six sides that can be permuted via rotations of the die. We note
that not all permutations can be realized as rotations and not all rotations
will give permutations. Let $G$ be the group of permutations that can be
realized as rotations. We shall consider $G$ as the symmetry group of the
die and observe that the uniform distribution on the six sides is the only
distribution that is invariant under the action of the symmetry group $G.$
\end{example}

\begin{example}
$G=\mathbb{R}/2\pi \mathbb{Z}$ is a commutative group that can be identified
with the group $SO\left( 2\right) $ of rotations in 2 dimensions. This is
the simplest example of a group that is compact but not finite.
\end{example}

For an object with symmetries the symmetry group defines a group action on
the object, and any group action on an object defines a symmetry group of
the object. A special case of a group action of the group $G$ is left
translation of the elements in $G$. Instead of studying distributions on
objects with symmetries, in this paper we shall focus on distributions on
the symmetry groups themselves. It is no serious restriction because a
distribution on the symmetry group of an object will induce a distribution
on the object itself.

Convergence of convolutions of probability measures were studied by
Stromberg \cite{Stromberg60} who proved weak convergence of convolutions of
probability measures. An information theoretic approach was introduced by
Csisz\'{a}r \cite{Csis64}. Classical methods involving characteristic
functions have been used to give conditions for uniform convergence of the
densities of convolutions \cite{Schlosman80}. See \cite{Johnson04} for a
review of the subject and further references.

Finally it is shown that convergence in information divergence corresponds
to uniform convergence of the rate distortion function and that weak
convergence corresponds to pointwise convergence of the rate distortion
function. In this paper we shall mainly consider convolutions as Markov
chains. This will give us a tool, which allows us to prove convergence of
iid. convolutions, and the rate of convergence is proved to be exponential.

The rest of the paper is organized as follows. In Section \ref{SecDistortion}
we establish a number of simple results on distortion functions on compact
set. These results will be used in Section \ref{SecRateDist}. In Section \ref%
{SecHaar} we define the uniform distribution on a compact group as the
uniquely determined Haar probability measures. In Section \ref{SecRateDist}
it is shown that the uniform distribution is the maximum entropy
distribution on a compact group in the sense that it maximizes the rate
distortion function at any positive distortion level. Convergence of
convolutions of a distribution to the uniform distribution is established in
Section \ref{SecConvergence} using Markov chain techniques, and the rate of
convergence is discussed in Section \ref{SecRateConv}. The group $SO\left(
2\right) $ is used as our running example. We finish with a short discussion.

\section{Distortion on compact groups\label{SecDistortion}}

Let $G$ be a compact group where $\ast $ denotes the composition. The
neutral element will be denoted $e$ and the inverse of the element $g$ will
be denoted $g^{-1}$.

We shall start with some general comments on distortion functions on compact
sets. Assume that the group both plays the role as source alphabet and
reproduction alphabet. A \emph{distortion function} $d:G\times G\rightarrow 
\mathbb{R}$ is given and we will assume that $d\left( x,y\right) \geq 0$
with equality if and only if $x=y.$ We will also assume that the distortion
function is continuous.

\begin{example}
As distortion function on $SO\left( 2\right) $ we use the squared Euclidean
distance between the corresponding points on the unit circle, i.e.%
\begin{eqnarray*}
d\left( x,y\right)  &=&4\sin ^{2}\left( \frac{x-y}{2}\right)  \\
&=&2-2\cos \left( x-y\right) .
\end{eqnarray*}%
This illustrated in Figure \ref{vinkler}.%
\begin{figure}[ptb]\begin{center}
\includegraphics[
natheight=16.9892in, natwidth=18.0002in, height=2.1162in, width=2.2416in]
{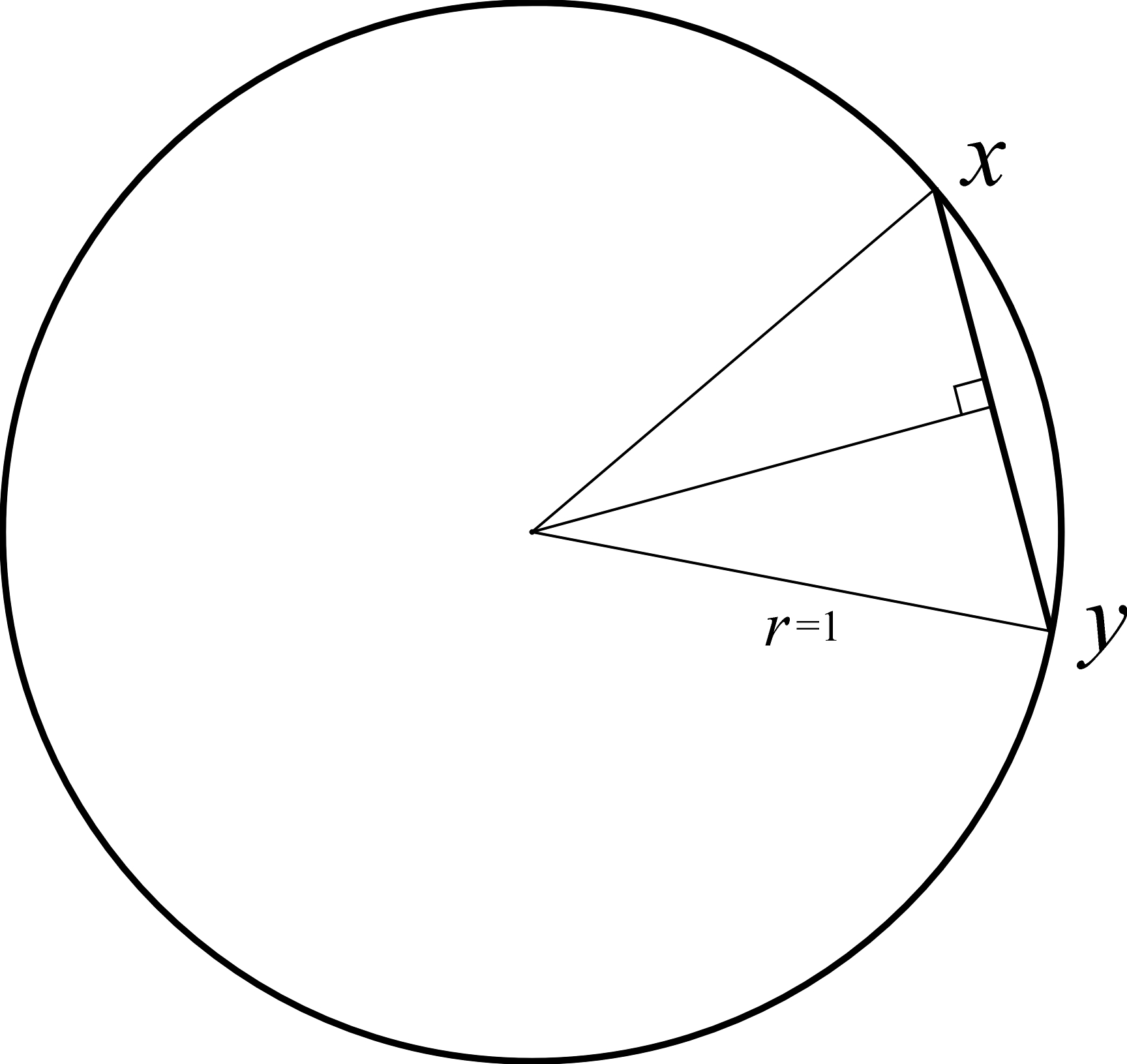}\caption{Squared Euclidean
distance between the rotation angles $x$ and $y.$}\label{vinkler}%
\end{center}\end{figure}%
\end{example}

The distortion function might be a metric but even if the distortion
function is not a metric, the relation between the distortion function and
the topology is the same as if it was a metric. One way of constructing a
distortion function on a group is to use the squared Hilbert-Smidt norm in a
unitary representation of the group.

\begin{theorem}
If $C$ is a compact set and $d:C\times C\rightarrow \mathbb{R}$ is a
non-negative continuous distortion function such that $d\left( x,y\right) =0$
if and only if $x=y,$ then the topology on $C$ is generated by the
distortion balls $\left\{ {x\in C\mid d\left( x,y\right) <r}\right\} $ where 
$y\in C$ and $r>0.$
\end{theorem}

\begin{proof}
We have to prove that a subset $B\subseteq C$ is open if and only if for any $%
y\in B$ there exists a ball that is a subset of $B$ and contains $y$. Assume
that $B\subset C$ is open and that $y\in B.$ Then $\complement B$ compact.
Hence, the function $x\rightarrow d\left( x,y\right) $ has a minimum $r$ on $%
\complement B$ and $r$ must be positive because $r=d\left( x,y\right) =0$
would imply that $x=y\in B.$ Therefore $\left\{ {x\in C\mid d\left(
x,y\right) <r}\right\} \subseteq B.$

Continuity of $d$ implies that the balls $\left\{ {x\in C\mid d\left(
x,y\right) <r}\right\} $ are open. If any point in $B$ is contained in an
open ball, then $B$ is a union of open set and open.
\end{proof}

The following theorem may be considered as a kind of uniform continuity of
the distortion function or as a substitute for the triangular inequality
when $d$ is not a metric.

\begin{lemma}
\label{LemmaUnif}If $C$ is a compact set and $d:C\times C\rightarrow \mathbb{%
R}$ is a non-negative continuous distortion function such that $d\left(
x,y\right) =0$ if and only if $x=y$, then there exists a continuous function 
$f_{1}$ satisfying $f_{1}\left( 0\right) =0$ such that 
\begin{equation}
\left\vert d\left( x,y\right) -d\left( z,y\right) \right\vert \leq
f_{1}\left( d\left( z,y\right) \right) \text{ for }x,y,z\in C.  \label{unif}
\end{equation}
\end{lemma}

\begin{proof}
Assume that the theorem does not hold. Then there exists $\epsilon >0$ and a
net $\left( x_{\lambda },y_{\lambda },z_{\lambda }\right) _{\lambda \in
\Lambda }$ such that 
\begin{equation*}
d\left( x_{\lambda },y_{\lambda }\right) -d\left( z_{\lambda },y_{\lambda
}\right) >\epsilon
\end{equation*}%
and $d\left( z_{\lambda },y_{\lambda }\right) \rightarrow 0.$ A net in a
compact set has a convergent subnet so without loss of generality we may
assume that the net $\left( x_{\lambda },y_{\lambda },z_{\lambda }\right)
_{\lambda \in \Lambda }$ converges to some triple $\left( x_{\infty
},y_{\infty },z_{\infty }\right) .$ By continuity of the distortion function
we get 
\begin{equation*}
d\left( x_{\infty },y_{\infty }\right) -d\left( z_{\infty },y_{\infty
}\right) \geq \epsilon
\end{equation*}%
and $d\left( z_{\infty },y_{\infty }\right) =0,$ which implies $z_{\infty
}=y_{\infty }$ and we have a contradiction.
\end{proof}

We note that if a distortion function satisfies (\ref{unif}) then it defines
a topology in which the distortion balls are open.

In order to define the weak topology on probability distributions we extend
the distortion function from $C\times C$ to $M_{+}^{1}\left( C\right) \times
M_{+}^{1}\left( C\right) $ via 
\begin{equation*}
d\left( P,Q\right) =\inf E\left[ \ d\left( X,Y\right) \right] ,
\end{equation*}%
where $X$ and $Y$ are random variables with values in $C$ and the infimum is
taken all joint distributions on $\left( X,Y\right) $ such that the marginal
distribution of $X$ is $P$ and the marginal distribution of $Y$ is $Q.$ The
distortion function is continuous so $\left( x,y\right) \rightarrow d\left(
x,y\right) $ has a maximum that we denote $d_{\max }.$

\begin{theorem}
If $G$ is a compact set and $d:C\times C\rightarrow \mathbb{R}$ is a
non-negative continuous distortion function such that $d\left( x,y\right) =0$
if and only if $x=y$, then 
\begin{equation*}
\left\vert d\left( P,Q\right) -d\left( S,Q\right) \right\vert \leq
f_{2}\left( d\left( S,P\right) \right) \text{ for }P,Q,S\in M_{+}^{1}\left(
C\right)
\end{equation*}%
for some continuous function $f_{2}$ satisfying $f_{2}\left( 0\right) =0.$
\end{theorem}

\begin{proof}
According to Lemma \ref{LemmaUnif} there exists a function $f_{1}$
satisfying (\ref{unif}). We use that 
\begin{align*}
E\left[ \left\vert d\left( X,Y\right) -d\left( Z,Y\right) \right\vert \right]
& \leq E\left[ f_{1}\left( d\left( Z,X\right) \right) \right] \\
& =E\left[ f_{1}\left( d\left( Z,X\right) \right) \mid d\left( Z,X\right)
\leq \delta \right] \cdot P\left( d\left( Z,X\right) \leq \delta \right) \\
& +E\left[ f_{1}\left( d\left( Z,X\right) \right) \mid d\left( Z,X\right)
>\delta \right] \cdot P\left( d\left( Z,X\right) >\delta \right) \\
& \leq f_{1}\left( \delta \right) \cdot 1+f_{1}\left( d_{\max }\right) \cdot 
\frac{E\left[ d\left( Z,X\right) \right] }{\delta } \\
& \leq f_{1}\left( \delta \right) +f_{1}\left( d_{\max }\right) \cdot \frac{%
d\left( S,P\right) }{\delta }.
\end{align*}%
This hold for all $\delta >0$ and in particular for $\delta =\left( d\left(
S,P\right) \right) ^{1/2}$, which proves the theorem.
\end{proof}

The theorem can be used to construct the \emph{weak topology} on $%
M_{+}^{1}\left( C\right) $ with 
\begin{equation*}
\left\{ P\in M_{+}^{1}\left( C\right) \mid d\left( P,Q\right) <r\right\} ,
\end{equation*}
$P\in M_{+}^{1}\left( C\right) ,r>0$ as open balls that generate the
topology. We note without proof that this definition is equivalent with the
quite different definition of weak topology that one will find in most
textbooks.

For a group $G$ we assume that the distortion function is \emph{right
invariant} in the sense that for all $x,y,z\in G$ a distortion function $d$
satisfies%
\begin{equation*}
d\left( x\ast z,y\ast z\right) =d\left( x,y\right) .
\end{equation*}%
A right invariant distortion function satisfies $d\left( x,y\right) =d\left(
x\ast y^{-1},e\right) $, so right invariant continuous distortion functions
of a group can be constructed from non-negative functions with a minimum in $%
e$.

\section{The Haar measure\label{SecHaar}}

We use $\ast $ to denote convolution of probability measures on $G.$ For $%
g\in G$ we shall use $g\ast P$ to denote the $g$-translation of the measure $%
P$ or, equivalently, the convolution with a measure concentrated in $g$. The 
$n$-fold convolution of a distribution $P$ with itself will be denoted $%
P^{\ast n}.$ For random variables with values in $G$ one can formulate an
analog of the central limit theorem. We recall some facts about probability
measures on compact groups and their \emph{Haar measures}.

\begin{definition}
Let $G$ be a group. A measure $P$ is said to be a \emph{left Haar measure}
if $g\ast P=P$ for any $g\in G$. Similarly, $P$ is said to be a \emph{right
Haar measure} if $P\ast g=P$ for any $g\in G.$ A measure is said to be a 
\emph{Haar measure} if it is both a left Haar measure and a right Haar
measure.
\end{definition}

\begin{example}
The uniform distribution on $SO\left( 2\right) $ or $\mathbb{R}/2\pi Z$ has
density $1/2\pi $ with respect to the Lebesgue measure on $\left[ 0;2\pi %
\right[ .$ The function 
\begin{equation}
f\left( x\right) =1+\sum_{n=1}^{\infty }a_{n}\cos \left( n\left( x+\phi
_{n}\right) \right)  \label{Fourier}
\end{equation}%
is a density on a probability distribution $P$ on $SO\left( 2\right) $ if
the Fourier coefficients $a_{n}$ are sufficiently small so that $f$ is
non-negative. A sufficient condition for $f$ to be non-negative is that $%
\sum_{n=1}^{\infty }\left\vert a_{n}\right\vert \leq 1.$

Translation by $y$ gives a distribution with density 
\begin{equation*}
f\left( x-y\right) =1+\sum_{n=1}^{\infty }a_{n}\cos \left( n\left( x-y+\phi
_{n}\right) \right) .
\end{equation*}%
The distribution $P$ is invariant if and only if $f$ is $1$ or,
equivalently, all Fourier coefficients $\left( a_{n}\right) _{n\in \mathbb{N}%
}$ are $0.$
\end{example}

A measure $P$ on $G$ is said to have \emph{full support} if the support of $%
P $ is $G,$ i.e. $P\left( A\right) >0$ for any non-empty open set $%
A\subseteq G.$ The following theorem is well-known \cite%
{Haar1933,Halmos1950,Conway90}.

\begin{theorem}
\label{Thm1}Let $U$ be a probability measure on the compact group $G.$ Then
the following four conditions are equivalent.

\begin{itemize}
\item $U$ is a left Haar measure.

\item $U$ is a right Haar measure.

\item $U$ has full support and is idempotent in the sense that $U\ast U=U.$

\item There exists a probability measure $P$ on $G$ with full support such
that $P\ast U=U.$

\item There exists a probability measure $P$ on $G$ with full support such
that $U\ast P=U.$
\end{itemize}

In particular a Haar probability measure is unique.
\end{theorem}

In \cite{Haar1933,Halmos1950,Conway90} one can find the proof that any
locally compact group has a Haar measure. The unique Haar probability
measure on a compact group will be called the \emph{uniform distribution}
and denoted $U.$ \newline
For probability measures $P$ and $Q$ the \emph{information divergence from} $%
P$ \emph{to} $Q$ is defined by 
\begin{equation*}
D\left( P\Vert Q\right) =\left\{ 
\begin{array}{cc}
\int \log \frac{dP}{dQ}~dP, & \text{if }P\ll Q; \\ 
\infty , & \text{otherwise.}%
\end{array}%
\right.
\end{equation*}%
We shall often calculate the divergence from a distribution to the uniform
distribution $U,$ and introduce the notation 
\begin{equation*}
D\left( P\right) =D\left( P\Vert U\right) .
\end{equation*}%
For a random variable $X$ with values in $G$ we will sometimes write $%
D\left( X\Vert U\right) $ instead of $D\left( P\Vert U\right) $ when $X$ has
distribution $P.$

\begin{example}
The distribution $P$ with density $f$ given by (\ref{Fourier}) has 
\begin{eqnarray*}
D\left( P\right) &=&\frac{1}{2\pi }\int_{0}^{2\pi }f\left( x\right) \log
\left( f\left( x\right) \right) ~dx \\
&\approx &\frac{1}{2\pi }\int_{0}^{2\pi }f\left( x\right) \left( f\left(
x\right) -1\right) ~dx \\
&=&\frac{1}{2}\sum_{n=1}^{\infty }a_{n}^{2}.
\end{eqnarray*}
\end{example}

Let $G$ be a compact group with uniform distribution $U$ and let $F$ be a
closed subgroup of $G.$ Then the subgroup has a Haar probability measure $%
U_{F}$ and 
\begin{equation}
D\left( U_{F}\right) =\log \left( \left[ G:F\right] \right)  \label{coset}
\end{equation}%
where $\left[ G:F\right] $ denotes the index of $F$ in $G.$ In particular $%
D\left( U_{F}\right) $ is finite if and only if $\left[ G:F\right] $ is
finite.

\section{The rate distortion theory\label{SecRateDist}}

We will develop aspects of the rate distortion theory of a compact group $G.$
Let $P$ be a probability measure on $G.$ We observe that compactness of $G$
implies that a covering of $G$ by distortion balls of radius $\delta >0$
contains a finite covering. If $k$ is the number of balls in a finite
covering then $R_{P}\left( \delta \right) \leq \log \left( k\right) $ where $%
R_{P}$ is the rate distortion function of the probability measure $P.$ In
particular the rate distortion function is upper bounded. The entropy of a
probability distribution $P$ is given by $H\left( P\right) =R_{P}\left(
0\right) $. If the group is finite then the uniform distribution maximizes
the Shannon entropy $R_{P}\left( 0\right) $ but if the group is not finite
then in principle there is no entropy maximizer. As we shall see the uniform
distribution still plays the role of entropy maximizer in the sense that the
uniform distribution maximize the value $R_{P}\left( \delta \right) $ of the
rate distortion function for any positive distortion level $\delta >0$. The
rate distortion function $R_{P}$ can be studied using its convex conjugate $%
R_{P}^{\ast }$ given by%
\begin{equation*}
R_{P}^{\ast }\left( \beta \right) =\sup_{\delta }\beta \cdot \delta
-R_{P}\left( \delta \right) .
\end{equation*}%
The rate distortion function is then recovered by the formula%
\begin{equation*}
R_{P}\left( \delta \right) =\sup_{\beta }\beta \cdot \delta -R_{P}^{\ast
}\left( \beta \right) .
\end{equation*}%
The techniques are pretty standard \cite{Vogel92}.

\begin{theorem}
\label{RateDistThm}The rate distortion function of the uniform distribution
is given by 
\begin{equation*}
R_{U}^{\ast }\left( \beta \right) =\log \left( Z\left( \beta \right) \right)
\end{equation*}%
where $Z$ is the partition function defined by%
\begin{equation*}
Z\left( \beta \right) =\int_{G}\exp \left( \beta \cdot d\left( g,e\right)
\right) ~dUg.
\end{equation*}%
The rate distortion function of an arbitrary distribution $P$ satisfies%
\begin{equation}
R_{U}-D\left( P\Vert U\right) \leq R_{P}\leq R_{U}.  \label{opned}
\end{equation}
\end{theorem}

\begin{proof}
First we prove a Shannon type lower bound on the rate distortion function of
an arbitrary distribution $P$ on the group. Let $X$ be a random variable
with values in $G$ and distribution $P$, and let $\hat{X}$ be a random
variable coupled with $X$ such that the mean distortion $E\left[ d\left( X,%
\hat{X}\right) \right] $ equals $\delta $. Then%
\begin{align}
I\left( X;\hat{X}\right) & =D\left( X\Vert U\mid \hat{X}\right) -D\left(
X\Vert U\right) \\
& =D\left( X\ast \hat{X}^{-1}\Vert U\mid \hat{X}\right) -D\left( X\Vert
U\right) \\
& \geq D\left( X\ast \hat{X}^{-1}\Vert U\right) -D\left( X\Vert U\right) .
\label{nedre}
\end{align}%
Now, $E\left[ d\left( X,\hat{X}\right) \right] =E\left[ d\left( X\ast \hat{X}%
^{-1},e\right) \right] $ and 
\begin{equation*}
D\left( X\ast \hat{X}^{-1}\Vert U\right) \geq D\left( P_{\beta }\Vert
U\right)
\end{equation*}%
where $P_{\beta }$ is the distribution that maximizes divergence under the
constraint $E\left[ d\left( Y,e\right) \right] =\delta $ when $Y$ has
distribution $P_{\beta }.$ The distribution $P_{\beta }$ is given by the
density%
\begin{equation*}
\frac{dP_{\beta }}{dU}\left( g\right) =\frac{\exp \left( \beta \cdot d\left(
g,e\right) \right) }{Z\left( \beta \right) }.
\end{equation*}%
where $\beta $ is determined by the condition $\delta =Z^{\prime }\left(
\beta \right) /Z\left( \beta \right) .$\newline
If $P$ is uniform then a joint distribution is obtained by choosing $\hat{X}$
uniformly distributed, and choosing $Y$ distributed according to $P_{\beta }$
and independent of $\hat{X}.$ Then $X=Y\ast \hat{X}$ is distributed
according to $P_{\beta }\ast U=U$, and we have equality in (\ref{nedre}).
Hence the rate determined the lower bound (\ref{nedre}) is achievable for
the uniform distribution, which prove the first part of the theorem, and the
left inequality in (\ref{opned}).\newline
The joint distribution on $\left( X,\hat{X}\right) $ that achieved the rate
distortion function when $X$ has a uniform distribution, defines a Markov
kernel $\Psi :X\rightarrow \hat{X}$ that is invariant under translations in
the group. For any distribution $P$ the joint distribution on $\left( X,\hat{%
X}\right) $ determined by $P$ and $\Psi $ gives an achievable pair of
distortion, and rate that is on the rate distortion curve of the uniform
distribution. This proves the right inequality in Equation (\ref{opned}).
\end{proof}

\begin{example}
For the group $SO\left( 2\right) $ the rate distortion function can be
parametrized using the modified Bessel functions $I_{j},j\in \mathbb{N}_{0}$%
. The partition function is given by%
\begin{align*}
Z\left( \beta \right) & =\int_{G}\exp \left( \beta \cdot d\left( g,e\right)
\right) ~dUg \\
& =\frac{1}{2\pi }\int_{0}^{2\pi }\exp \left( \beta \cdot \left( 2-2\cos
x\right) \right) ~dx \\
& =\exp \left( 2\beta \right) \cdot \frac{1}{\pi }\int_{0}^{\pi }\exp \left(
-2\beta \cdot \cos x\right) ~dx \\
& =\exp \left( 2\beta \right) \cdot I_{0}\left( -2\beta \right) .
\end{align*}%
Hence $R_{U}^{\ast }\left( \beta \right) =$ $\log \left( Z\left( \beta
\right) \right) =2\beta +\log \left( I_{0}\left( -2\beta \right) \right) $.
The distortion $\delta $ corresponding to $\beta $ is given by 
\begin{equation*}
\delta =2-2\frac{I_{1}\left( -2\beta \right) }{I_{0}\left( -2\beta \right) }
\end{equation*}%
and the corresponding rate is 
\begin{eqnarray*}
R_{U}\left( \delta \right)  &=&\beta \cdot \delta -\left( 2\beta +\log
\left( I_{0}\left( -2\beta \right) \right) \right)  \\
&=&-\beta \cdot 2\frac{I_{1}\left( -2\beta \right) }{I_{0}\left( -2\beta
\right) }-\log \left( I_{0}\left( -2\beta \right) \right) .
\end{eqnarray*}%
These joint values of distortion and rate can be plotted with $\beta $ as
parameter as illustrated in Figure \ref{Bessel}.%
\begin{figure}[ptb]\begin{center}
\includegraphics[
natheight=12.9272in, natwidth=15.4585in, height=3.6832in, width=4.4002in]
{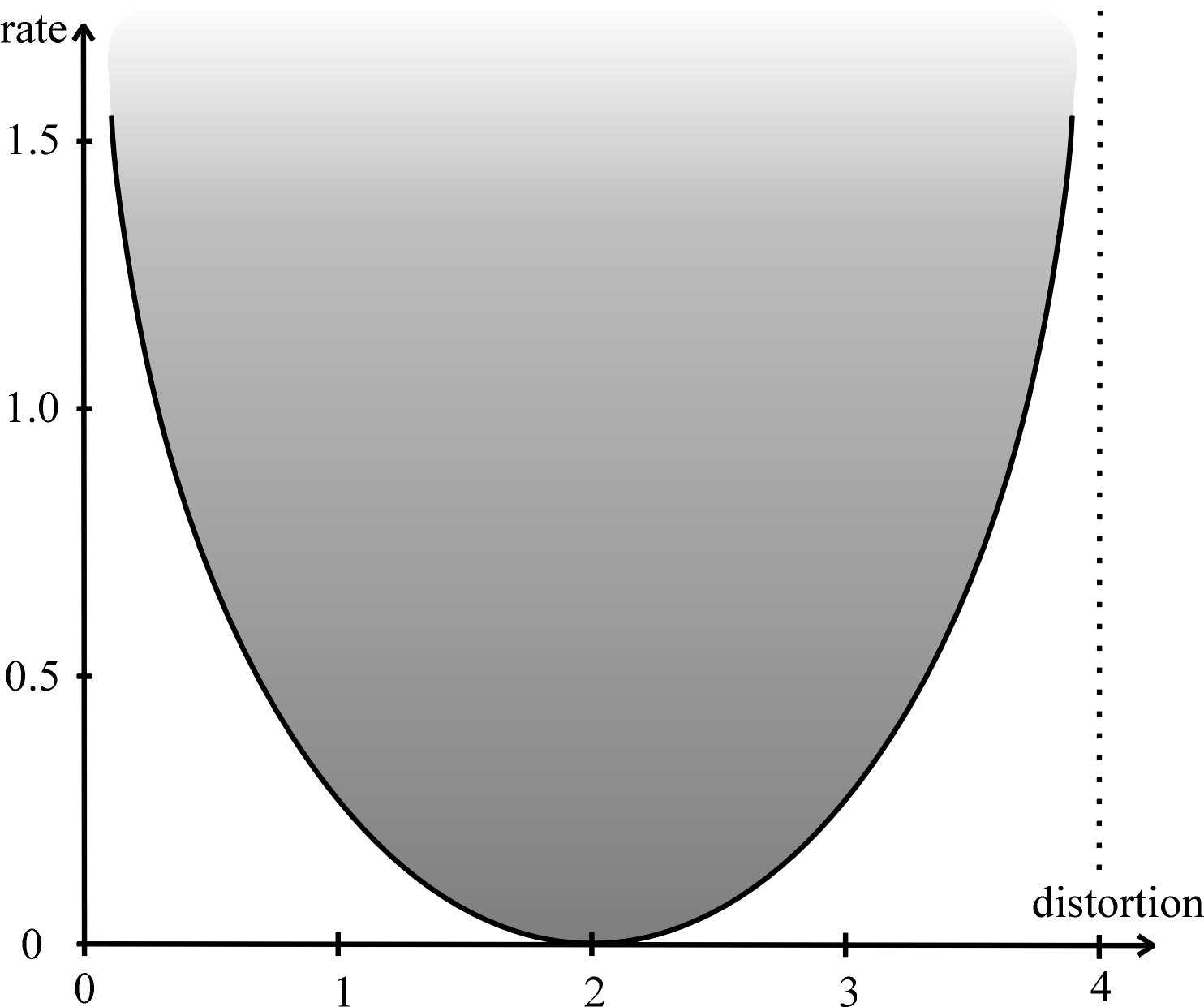}\caption{The rate distortion
region of the uniform distribution on $SO\left( 2\right) $ is shaded. The
rate distortion function is the lower bounding curve. In the figure the rate
is measured in nats. The critical distortion $d_{crit}$ equals 2, and the
dashed line indicates $d_{\max }=4.$}\label{Bessel}%
\end{center}\end{figure}%
\end{example}

The minimal rate of the uniform distribution is achieved when $X$ and $\hat{X%
}$ are independent. In this case the distortion is $E\left[ d\left( X,\hat{X}%
\right) \right] =\int_{G}d\left( x,e\right) ~dPx.$ This distortion level
will be called the critical distortion and will be denoted $d_{crit}.$ On
the interval $\left] 0;d_{crit}\right] $ the rate distortion function is
decreasing and the distortion rate function is the inverse $R_{P}^{-1}$ of
the rate distortion function $R_{P}$ on this interval. The distortion rate
function satisfies:

\begin{theorem}
\label{inverse}The distortion rate function of an arbitrary distribution $P$
satisfies%
\begin{equation}
R_{U}^{-1}\left( \delta \right) -f_{2}\left( d\left( P,U\right) \right) \leq
R_{P}^{-1}\left( \delta \right) \leq R_{U}^{-1}\left( \delta \right) ~\text{%
for }\delta \leq d_{crit}
\end{equation}%
for some increasing continuous function $f_{2}$ satisfying $f_{2}\left(
0\right) =0.$
\end{theorem}

\begin{proof}
The right hand side follows because $R_{U}$ is decreasing in the interval $%
\left[ 0;d _{crit}\right] $ Let $X$ be a random variable with
distribution $P$ and let $Y$ be a random variable coupled with $X.$ Let $Z$
be a random variable coupled with $X$ such that $E\left[ d\left( X,Z\right) %
\right] =d\left( P,U\right) .$ The couplings between $X$ and $Y$, and
between $X$ and $Z$ can be extended to a joint distribution on $X,Y$ and $Z$
such that $Y$ and $Z$ are independent given $X.$ For this joint distribution
we have 
\begin{equation*}
I\left( Z;Y\right) \leq I\left( X,Y\right)
\end{equation*}%
and 
\begin{equation*}
\left\vert E\left[ d\left( Z,Y\right) \right] -E\left[ d\left( X,Y\right) %
\right] \right\vert \leq f_{2}\left( d\left( P,U\right) \right) .
\end{equation*}%
We have to prove that 
\begin{equation*}
E\left[ d\left( X,Y\right) \right] \geq R_{U}^{-1}\left( I\left( X,Y\right)
\right) -f_{2}\left( d\left( P,U\right) \right)
\end{equation*}%
but $I\left( Z;Y\right) \leq I\left( X,Y\right) $ so it is sufficient to
prove that 
\begin{equation*}
E\left[ d\left( X,Y\right) \right] \geq R_{U}^{-1}\left( I\left( Z,Y\right)
\right) -f_{2}\left( d\left( P,U\right) \right)
\end{equation*}%
and this follows because $E\left[ d\left( Z,Y\right) \right] \geq
R_{U}^{-1}\left( I\left( Z,Y\right) \right) . $
\end{proof}

\section{Convergence of convolutions\label{SecConvergence}}

We shall prove that under certain conditions the $n$-fold convolutions $%
P^{\ast n}$ converge to the uniform distribution.

\begin{example}
The function 
\begin{equation*}
f\left( x\right) =1+\sum_{n=1}^{\infty }a_{n}\cos \left( n\left( x+\phi
_{n}\right) \right)
\end{equation*}%
is a density on a probability distribution $P$ on $G$ if the Fourier
coefficients $a_{n}$ are sufficiently small. If $\left( a_{n}\right) $ and $%
\left( b_{n}\right) $ are Fourier coefficients of $P$ and $Q$ then the
convolution has density%
\begin{multline*}
\frac{1}{2\pi }\int_{0}^{2\pi }\left( 1+\sum_{n=1}^{\infty }a_{n}\cos
n\left( x-y+\phi _{n}\right) \right) \left( 1+\sum_{n=1}^{\infty }b_{n}\cos
n\left( y+\psi _{n}\right) \right) ~dy \\
=1+\frac{1}{2\pi }\sum_{n=1}^{\infty }\int_{0}^{2\pi }a_{n}b_{n}\cos n\left(
x-y+\phi _{n}\right) \cos n\left( y+\psi _{n}\right) ~dy \\
=1+\frac{1}{2\pi }\sum_{n=1}^{\infty }\int_{0}^{2\pi }a_{n}b_{n}\cos \left(
n\left( x+\phi _{n}+\psi _{n}\right) -ny\right) \cos \left( ny\right) ~dy \\
=1+\frac{1}{2\pi }\sum_{n=1}^{\infty }\int_{0}^{2\pi }a_{n}b_{n}\left( 
\begin{array}{c}
\cos n\left( x+\phi _{n}+\psi _{n}\right) \cos \left( ny\right) \\ 
+\sin \left( n\left( x+\phi _{n}+\psi _{n}\right) \right) \sin \left(
ny\right)%
\end{array}%
\right) \cos \left( ny\right) ~dy \\
=1+\sum_{n=1}^{\infty }\frac{a_{n}b_{n}\cos \left( n\left( x+\phi _{n}+\psi
_{n}\right) \right) }{2\pi }\int_{0}^{2\pi }\cos ^{2}\left( ny\right) ~dy \\
=1+\sum_{n=1}^{\infty }\frac{a_{n}b_{n}\cos \left( n\left( x+\phi _{n}+\psi
_{n}\right) \right) }{2}.
\end{multline*}%
Therefore the $n$-fold convolution has density%
\begin{equation*}
1+\sum_{k=1}^{\infty }\frac{a_{k}^{n}\cos \left( k\left( x+n\phi _{k}\right)
\right) }{2^{n-1}}=1+\sum_{k=1}^{\infty }\left( \frac{a_{k}}{2}\right)
^{n}2\cos \left( k\left( x+n\phi _{k}\right) \right) .
\end{equation*}%
Therefore each of the Fourier coefficients is exponentially decreasing.
\end{example}

Clearly, if $P$ is uniform on a proper subgroup then convergence does not
hold. In several papers on this topic \cite[and references therein]%
{Johnson2000, Johnson04} it is claimed and \textquotedblleft
proved\textquotedblright\ that if convergence does not hold then the support
of $P$ is contained in the coset of a proper normal subgroup. The proofs
therefore contain errors that seem to have been copied from paper to paper.
To avoid this problem and make this paper more self-contained we shall
reformulate and reprove some already known theorems. \newline
In the theory of finite Markov chains is well-known that there exists an
invariant probability measure. Certain Markov chains exhibits periodic
behavior where a certain distribution is repeated after a number of
transitions. All distributions in such a cycle will lie at a fixed distance
from any (fixed) measure, where the distance is given by information
divergence or total variation (or any other Csisz{\'{a}}r $f$-divergence).
It is also well-known that finite Markov chains without periodic behavior
are convergent. In general a Markov chain will converge to a
\textquotedblleft cyclic\textquotedblright\ behavior as stated in the
following theorem \cite{Harremoes2009}.

\begin{theorem}
\label{main}Let $\Phi $ be a transition operator on a state space $A$ with
an invariant probability measure $Q_{in}.$ If $D\left( S\parallel Q\right)
<\infty $ then there exists a probability measure $P^{\ast }$ such that $%
D\left( \Phi ^{n}S\parallel \Phi ^{n}Q\right) \rightarrow 0$ and $D\left(
\Phi ^{n}Q\parallel Q_{in}\right) $ is constant.
\end{theorem}

We shall also use the following proposition that has a purely computational
proof \cite{Topsoe67}.

\begin{proposition}
Let $P_{x}, x\in X$ be distributions and let $Q$ be a probability
distribution on $X.$ Then 
\begin{equation*}
\int D\left( P_{x}\parallel Q\right) ~dQx =D\left( \int P_{x}dQx\parallel
Q\right) +\int D\left( P_{x}\parallel\int P_{t} ~dQt \right) ~dQx .
\end{equation*}
\end{proposition}

We denote the set of probability measures on $G$ by $M_{+}^{1}\left( G
\right)$.

\begin{theorem}
\label{konvergens}Let $P$ be a distribution on a compact group $G$ and
assume that the support of $P$ is not contained in any nontrivial coset of a
subgroup of $G.$ Then, if $D\left( S\Vert U\right) $ is finite then $D\left(
P^{\ast n}\ast S\Vert U\right) \rightarrow 0$ for $n\rightarrow \infty .$
\end{theorem}

\begin{proof}
Let $\Psi :G\rightarrow M_{+}^{1}\left( G\right) $ denote the Markov kernel $%
\Psi \left( g\right) =P\ast g.$ Then $P^{\ast n}\ast S=\Psi ^{n}\left( P\ast
S\right) .$ Thus there exists a probability measure $Q$ on $G$ such that $%
D\left( \Psi ^{n}\left( P\right) \Vert \Psi ^{n}\left( Q\right) \right)
\rightarrow 0$ for $n\rightarrow \infty $ and such that $D\left( \Psi
^{n}\left( Q\right) \right) $ is constant. We shall prove that $Q=U.$

First we note that%
\begin{align*}
D\left( Q\right) & =D\left( P\ast Q\right) \\
& =\int_{G}\left( D\left( g\ast Q\right) -D\left( g\ast Q\Vert P\ast
Q\right) \right) ~dPg \\
& =D\left( Q\right) -\int_{G}D\left( g\ast Q\Vert P\ast Q\right) ~dPg\ .
\end{align*}%
Therefore $g\ast Q=P\ast Q$ for $P$ almost every $g\in G.$ Thus there exists
at least one $g_{0}\in G$ such that $g_{0}\ast Q=P\ast Q.$ Then $Q=\tilde{P}%
\ast Q$ where $\tilde{P}=g_{0}^{-1}\ast P.$ \newline
Let $\tilde{\Psi}:G\rightarrow M_{+}^{1}\left( G\right) $ denote the Markov
kernel $g\rightarrow \tilde{P}\ast g.$ Put%
\begin{equation*}
P_{n}=\frac{1}{n}\sum_{i=1}^{n}\tilde{P}^{\ast i}=\frac{1}{n}\sum_{i=1}^{n}%
\tilde{\Psi}^{i-1}\left( \tilde{P}\right) .
\end{equation*}%
According to \cite{Harremoes2009} this ergodic mean will converge to a
distribution $T$ such that $\tilde{\Psi}\left( T\right) =T$ so that $T\ast 
\tilde{P}=T.$ Hence we also have that $T\ast T=T,$ i.e. $T$ is idempotent
and therefore supported by a subgroup of $G$. We know that $\tilde{P}$ is
not contained in any nontrivial subgroup of $G$ so the support of $T$ must
be $G$. We also get $Q=T\ast Q,$ which together with Theorem \ref{Thm1}
implies that $Q=U.$
\end{proof}

by choosing $S=P$ we get the following corollary.

\begin{corollary}
\label{divkonv}Let $P$ be a probability measure on the compact group $G$
with Haar probability measure $U$. Assume that the support of $P$ is not
contained in any coset of a proper subgroup of $G$ and $D\left( P\Vert
U\right) $ is finite. Then $D\left( P^{\ast n}\Vert U\right) \rightarrow 0$
for $n\rightarrow \infty $.
\end{corollary}

Corollary \ref{divkonv} together with Theorem \ref{RateDistThm} implies the
following result.

\begin{corollary}
Let $P$ be a probability measure on the compact group $G$ with Haar
probability measure $U$. Assume that the support of $P$ is not contained in
any coset of a proper subgroup of $G$ and $D\left( P\Vert U\right) $ is
finite. Then the rate distortion function of $P^{\ast n}$ converges
uniformly to the rate distortion function of the uniform distribution.
\end{corollary}

We also get weak versions of these results.

\begin{corollary}
\label{dweakkonv}Let $P$ be a probability measure on the compact group $G$
with Haar probability measure $U.$ Assume that the support of $P$ is not
contained in any coset of a proper subgroup of $G.$ Then $P^{\ast n}$
converges to $U$ in the weak topology, i.e. $d\left( P^{\ast n},U\right)
\rightarrow 0$ for $n\rightarrow \infty .$
\end{corollary}

\begin{proof}
If we take $S=P_{\beta }$ then $D\left( P_{\beta }\right) $ is finite and $%
D\left( P^{\ast n}\ast P_{\beta }\Vert U\right) \rightarrow 0$ for $%
n\rightarrow \infty $. We have 
\begin{eqnarray*}
d\left( P^{\ast n}\ast P_{\beta },U\right) &\leq &d_{\max }\left\Vert
P^{\ast n}\ast P_{\beta }-U\right\Vert \\
&\leq &d_{\max }\left( 2D\left( P^{\ast n}\ast P_{\beta }\Vert U\right)
\right) ^{1/2}
\end{eqnarray*}%
implying that $d\left( P^{\ast n}\ast P_{\beta },U\right) \rightarrow 0$ for 
$n\rightarrow \infty $. Now 
\begin{eqnarray*}
\left\vert d\left( P^{\ast n},U\right) -d\left( P^{\ast n}\ast P_{\beta
},U\right) \right\vert &\leq &f_{2}\left( d\left( P^{\ast n}\ast P_{\beta
},P^{\ast n}\right) \right) \\
&\leq &f_{2}\left( d\left( P_{\beta },e\right) \right) .
\end{eqnarray*}%
Therefore $\lim_{n\rightarrow \infty }\sup d\left( P^{\ast n},U\right) \leq
f_{2}\left( d\left( P_{\beta },e\right) \right) $ for all $\beta $, which
implies that 
\begin{equation*}
\lim_{n\rightarrow \infty }\sup d\left( P^{\ast n},U\right) =0.\qedhere
\end{equation*}
\end{proof}

\begin{corollary}
\label{pointwisekonv}Let $P$ be a probability measure on the compact group $%
G $ with Haar probability measure $U.$ Assume that the support of $P$ is not
contained in any coset of a proper subgroup of $G$ and $D\left( P\Vert
U\right) $ is finite. Then $R_{P^{\ast n}}$ converges to $R_{U}$ pointwise
on the interval $\left] 0;d_{\max }\right[ $ for $n\rightarrow \infty .$
\end{corollary}

\begin{proof}
Corollary \ref{dweakkonv} together with Theorem \ref{inverse} implies
uniform convergence of the distortion rate function for distortion less than 
$d_{crit}$. This implies pointwise convergence of the rate distortion
function on $\left] 0;d_{crit}\right[ $ because rate distortion functions
are convex functions. The same argument works in the interval $\left]
d_{crit};d_{\max }\right[ .$ Pointwise convergence in $d_{crit}$ must also
hold because of continuity.
\end{proof}

\section{Rate of convergence\label{SecRateConv}}

Normally the rate of convergence will be exponential. If the density is
lower bounded this is well-known. We bring a simplified proof of this.

\begin{lemma}
\label{lower}Let $P$ be a probability distribution on the compact group $G$
with Haar probability measure $U.$ If $dP/dU\geq c > 0 $ and $D\left(
P\right) $ is finite, then%
\begin{equation*}
D\left( P^{^n}\right) \leq\left( 1-c\right) ^{n-1}D\left( P\right) .
\end{equation*}
\end{lemma}

\begin{proof}
First we write%
\begin{equation*}
P=\left( 1-c\right) \cdot S+c\cdot U
\end{equation*}%
where $S$ denotes the probability measure

\begin{equation*}
S=\frac{P-cU}{1-c}.
\end{equation*}

For any distribution $Q$ on $G$ we have%
\begin{align*}
D\left( Q\ast P\right) & =D\left( \left( 1-c\right) \cdot Q\ast S+c\cdot
Q\ast U\right) \\
& \leq \left( 1-c\right) \cdot D\left( Q\ast S\right) +c\cdot D\left( Q\ast
U\right) \\
& \leq \left( 1-c\right) \cdot D\left( Q\right) +c\cdot D\left( U\right) \\
& =\left( 1-c\right) \cdot D\left( Q\right) .
\end{align*}%
Here we have used convexity of divergence.
\end{proof}

If a distribution $P$ has support in a proper subgroup $F$ then%
\begin{align*}
D\left( P\right) & \geq D\left( U_{F}\right) \\
& =\log \left( \left[ G:F\right] \right) \\
& \geq \log \left( 2\right) =\text{1 bit}.
\end{align*}%
Therefore $D\left( P\right) <1$ bit implies that $P$ cannot be supported by
a proper subgroup, but it implies more.

\begin{proposition}
\label{1bit} If $P$ is a distribution on the compact group $G$ and $D\left(
P\right) <1$\textrm{\ bit} then $\frac{d\left( P\ast P\right) }{dU}$ is
lower bounded by a positive constant.
\end{proposition}

\begin{proof}
The condition $D\left( P\right) <1$\textrm{\ bit} implies that $U\left\{ 
\frac{dP}{dU}>0\right\} >1/2.$ Hence there exists $\varepsilon>0$ such that $%
U\left\{ \frac{dP}{dU}>\varepsilon\right\} >1/2.$ We have 
\begin{align*}
\frac{d\left( P\ast P\right) }{dU}\left( y\right) & =\int_{G}\frac {dP}{dU}%
\left( x\right) \cdot\frac{dP}{dU}\left( y-x\right) ~dUx \\
& \geq\int_{\left\{ \frac{dP}{dU}>\varepsilon\right\} }\varepsilon\cdot 
\frac{dP}{dU}\left( y-x\right) ~dUx \\
& \geq\varepsilon\cdot\int_{\left\{ \frac{dP}{dU}\left( x\right)
>\varepsilon\right\} \cap\left\{ \frac{dP}{dU}\left( y-x\right)
>\varepsilon\right\} }\varepsilon~dUx \\
& =\varepsilon^{2}\cdot U\left( \left\{ \frac{dP}{dU}\left( x\right)
>\varepsilon\right\} \cap\left\{ \frac{dP}{dU}\left( y-x\right)
>\varepsilon\right\} \right) .
\end{align*}
Using the inclusion-exclusion inequalities we get%
\begin{multline*}
U\left( \left\{ \frac{dP}{dU}\left( x\right) >\varepsilon\right\}
\cap\left\{ \frac{dP}{dU}\left( y-x\right) >\varepsilon\right\} \right) \\
=U\left\{ \frac{dP}{dU}\left( x\right) >\varepsilon\right\} +U\left\{ \frac{%
dP}{dU}\left( y-x\right) >\varepsilon\right\}-U\left( \left\{ \frac{dP}{dU}%
\left( x\right) >\varepsilon\right\} \cup\left\{ \frac{dP}{dU}\left(
y-x\right) >\varepsilon\right\} \right) \\
\geq 2\cdot U\left\{ \frac{dP}{dU}\left( x\right) >\varepsilon\right\} -1.
\end{multline*}
Hence 
\begin{equation*}
\frac{d\left( P\ast P\right) }{dU}\left( y\right) \geq2\varepsilon
^{2}\left( U\left\{ \frac{dP}{dU}\left( x\right) >\varepsilon\right\}
-1/2\right)
\end{equation*}
for all $y\in G.$
\end{proof}

Combining Theorem \ref{konvergens}, Lemma \ref{lower}, and Proposition \ref%
{1bit} we get the following result.

\begin{theorem}
Let $P$ be a probability measure on a compact group $G$ with Haar
probability measure $U.$ If the support of $P$ is not contained in any coset
of a proper subgroup of $G$ and $D\left( P\right \Vert U) $ is finite then
the rate of convergence of $D\left( P^{\ast n}\right \Vert U) $ to zero is
exponential.
\end{theorem}

As a corollary we get the following result that was first proved by Kloss 
\cite{Kloss1959} for total variation.

\begin{corollary}
Let $P$ be a probability measure on the compact group $G$ with Haar
probability measure $U.$ If the support of $P$ is not contained in any coset
of a proper subgroup of $G$ and $D\left( P\Vert U\right) $ is finite then $%
P^{\ast n}$ converges to $U$ in variation and the rate of convergence is
exponential.
\end{corollary}

\begin{proof}
This follows directly from Pinsker's inequality \cite{Csiszar67,
Fedotovetal03}%
\begin{equation*}
\frac{1}{2}\left\Vert P^{\ast n}-U\right\Vert ^{2}\leq D\left( P^{\ast
n}\Vert U\right) .\qedhere
\end{equation*}
\end{proof}

\begin{corollary}
Let $P$ be a probability measure on the compact group $G$ with Haar
probability measure $U.$ If the support of $P$ is not contained in any coset
of a proper subgroup of $G$ and $D\left( P\Vert U\right) $ is finite, then
the density%
\begin{equation*}
\frac{dP^{\ast n}}{dU}
\end{equation*}%
converges to 1 point wise almost surely for $n$ tending to infinity.
\end{corollary}

\begin{proof}
The variation norm can be written as%
\begin{equation*}
\left\Vert P^{\ast n}-U\right\Vert =\int_{G}\left\vert \frac{dP^{\ast n}}{dU}%
-1\right\vert ~dU.
\end{equation*}
Thus%
\begin{equation*}
U\left( \left\vert \frac{dP^{\ast n}}{dU}-1\right\vert \geq\varepsilon
\right) \leq\frac{\left\Vert P^{\ast n}-U\right\Vert }{\varepsilon}.
\end{equation*}
The result follows by the exponential rate of convergence of $P^{\ast n}$ to 
$U$ in total variation combined with the Borel-Cantelli Lemma.
\end{proof}

\section{Discussion}

In this paper we have assumed the existence of the Haar measure by referring
to the literature. With the Haar measure we have then proved convergence of
convolutions using Markov chain techniques. The Markov chain approach can
also be used to prove the existence of the Haar measure by simply referring
to the fact that a homogenous Markov chain on a compact set has an invariant
distribution. The problem about this approach is that the proof that a
Markov chain on a compact set has an invariant distribution is not easier
than the proof of the existence of the Haar measure and is less known.

We have shown that the Haar probability measure maximizes the rate
distortion function at any distortion level. The normal proofs of the
existence of the Haar measure use a kind of covering argument that is very
close to the techniques found in rate distortion technique. There is a
chance that one can get an information theoretic proof of the existence of
the Haar measure. It seems obvious to use concavity arguments as one would
do for Shannon entropy but, as proved by Ahlswede \cite{Ahlswede1990a}, the
rate distortion function at a given distortion level is not a concave
function of the underlying distribution, so some more refined technique is
needed. 

As noted in the introduction for any algebraic structure $A$ the group $%
Aut\left( A\right) $ can be considered as symmetry group, it it has a
compact subgroup for which the results of this paper applies. It would be
interesting to extend the information theoretic approach to the algebraic
object $A$ itself, but in general there is no known equivalent to the Haar
measure for other algebraic structures. Algebraic structures are used
extensively in channel coding theory and cryptography so although the theory
may become more involved extensions of the result presented in this paper
are definitely worthwhile.

\section*{ Acknowledgement}

The author want to thank Ioannis Kontoyiannis for stimulating discussions.

\bibliographystyle{mdpi}
\bibliography{database}


\end{document}